\documentclass[a4paper,UKenglish,cleveref,autoref,thm-restate]{lipics-v2021}
\usepackage[utf8]{inputenc}
\usepackage{xspace}
\usepackage{graphicx}
\usepackage{amsthm,amssymb}
\usepackage[dvipsnames]{xcolor}
\usepackage{enumitem}
\usepackage{url}
\usepackage{hyperref}

\newtheorem{problem}[theorem]{Problem}

\newcommand{\etal}{{\it{et al.}}\xspace}
\newcommand{\R}{\ensuremath{\mathbb R}}

\newcommand{\nice}{{smooth}\xspace}

\newcommand{\RR}{\ensuremath{\mathcal R}\xspace}
\def\from{\mathrel\subset\mkern-10mu\joinrel\sim}

\newcommand{\myremark}[4]{\textcolor{blue}{\textsc{#1 #2:}} \textcolor{#4}{\textsf{#3}}}
\renewcommand{\myremark}[4]{}

\newcommand{\ben}    [2][says]{\myremark{Ben}    {#1}{#2}{JungleGreen}}
\newcommand{\maarten}[2][says]{\myremark{Maarten}{#1}{#2}{WildStrawberry}}

\newcommand{\remove}[1]{}

\title {Preprocessing Uncertain Data into Supersequences for Sorting and Gaps}

\author{Maarten L\"offler}
    {Department of Information and Computing Sciences; Utrecht University, the Netherlands}
    {m.loffler@uu.nl}%
    {https://orcid.org/0009-0001-9403-8856}%
    {}

   \author{Benjamin Raichel}{Department of Computer Science;
      University of Texas at Dallas, USA \and
      \url{http://utdallas.edu/\string~benjamin.raichel} }
   {benjamin.raichel@utdallas.edu}%
   {{https://orcid.org/0000-0001-6584-4843}}%
   {Work on this paper was partially supported by NSF CAREER Award
      1750780 and NSF Award 2311179.}

\authorrunning{M. L\"offler and B. Raichel} 

\Copyright{Maarten L\"offler and Ben Raichel} 



\category{} 

\relatedversion{} 

\acknowledgements{
The authors would like to thank anonymous reviewers of earlier versions of this work for their valuable and constructive comments.
}

\nolinenumbers 

\EventEditors{John Q. Open and Joan R. Access}
\EventNoEds{2}
\EventLongTitle{42nd Conference on Very Important Topics (CVIT 2016)}
\EventShortTitle{CVIT 2016}
\EventAcronym{CVIT}
\EventYear{2016}
\EventDate{December 24--27, 2016}
\EventLocation{Little Whinging, United Kingdom}
\EventLogo{}
\SeriesVolume{42}
\ArticleNo{23}

\begin{document}

\maketitle

\begin{abstract}
In the preprocessing framework for dealing with uncertain data, one is given a set of regions that one is allowed to preprocess to create some auxiliary structure such that when a realization of these regions is given, consisting of one point per region, this auxiliary structure can be used to reconstruct some desired output structure more efficiently than would have been possible without preprocessing. The framework has been successfully applied to several, mostly geometric, computational problems.

In this note, we propose using a supersequence of input items as the auxiliary structure, and explore its potential on the problems of sorting and computing the smallest or largest gap in a set of numbers.
That is, our uncertainty regions are intervals on the real line, and in the preprocessing phase we output a supersequence of the intervals such that the sorted order / smallest gap / largest gap of any realization is a subsequence of this sequence. 

We argue that supersequences are simpler than specialized auxiliary structures developed in previous work.
An advantage of using supersequences as the auxiliary structures is that it allows us to decouple the preprocessing phase from the reconstruction phase in a stronger sense than was possible in previous work, resulting in two separate algorithmic problems for which different solutions may be combined to obtain known and new results. We identify one key open problem which we believe is of independent interest.
\end{abstract}


\section{Introduction}

\subsection{Preprocessing framework} 
The {\em preprocessing framework} for dealing with data uncertainty was initially proposed by Held and Mitchell \cite{held2008triangulating} in the context of triangulating a point set in the plane. In this framework, we have a set $\RR = \{R_1, R_2, \ldots, R_n\}$ of {\em regions}, often in $\R^2$, and a point set $P = \{p_1, p_2, \ldots, p_n\}$ with $p_i \in R_i$
(we also write $P \from \RR$).
This model has two consecutive phases: a preprocessing phase, followed by a reconstruction phase. In the preprocessing phase we have access only to $\RR$ and we typically want to preprocess $\RR$ in $O(n \log n)$ time to create some linear-size auxiliary data structure which we will denote by $\Xi$. In the reconstruction phase, we have access to $P$ and we want to construct a desired output structure $S(P)$ on $P$ using $\Xi$ faster than would be possible otherwise.
Figure~\ref {fig:intro-prep} illustrates the model for the problem of {\em sorting}.

L{\"o}ffler and Snoeyink~\cite{loffler2010delaunay} were the first to use this model as a way to deal with data uncertainty: one may interpret the regions $\RR$ as {\em imprecise} points, and the points in $P$ as their true (initially unknown) locations. 
This interpretation makes sense in settings where approximate locations are cheap to obtain, but precise locations are costly or time-consuming, or when precise locations are not available at all but a large number of samples are taken from the distribution of all possible realizations for e.g. statistical analysis.
Various problems in computational geometry have been revisited in this framework~\cite {buchin2011delaunay,devillers2011delaunay,ezra2013convex,loffler2013unions,van2010preprocessing}.

\begin {figure}
 \begin{center}
  \includegraphics[width=\textwidth]{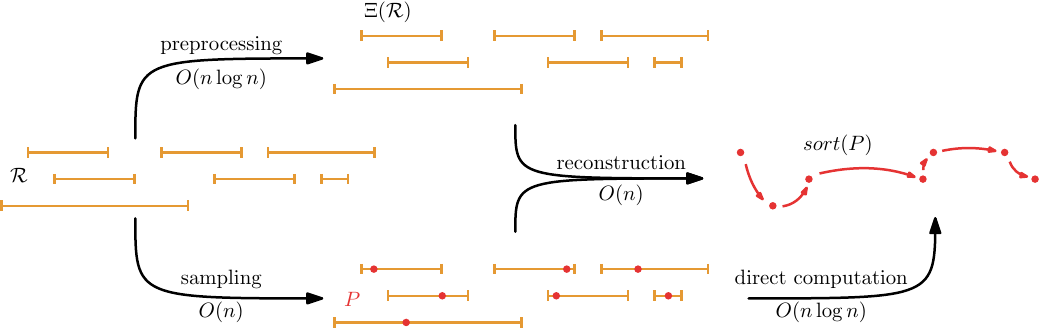}
 \end{center}
 \caption {A set of intervals $\RR$ of constant {\em ply} can be preprocessed into an auxiliary structure $\Xi(\RR)$ in $O(n\log n)$ time, such that the sorted order of a set of points $P$ that respects $\RR$ can be computed in linear time using $\Xi(\RR)$ (compared to $\Theta(n \log n)$ time without preprocessing)~\cite {BLMM11,held2008triangulating,hkls-paip-19}.}
 \label {fig:intro-prep}
\end {figure}

Arguably, many results using this framework are somewhat complicated. Part of the reason is the need for two algorithmic phases and the existence of the auxiliary data structure $\Xi$, which is something completely different in each paper and often does not have a clear intrinsic value. This makes it hard or impossible to separate the preprocessing phase from the reconstruction phase; they only make sense when viewed as a whole.

In this work, we explore a more restricted class of auxiliary structure: those where $\Xi$ is a {\em sequence} of elements from $\RR$ (possibly containing duplicated or emitted elements).
In principle, it is applicable to any computational problem where the output structure $S(P)$ is an ordered subset of $P$; that is, $S(P) = \langle p_{i_1}, p_{i_2}, \ldots, p_{i_s} \rangle$ for some indices $\{i_1, i_2, \ldots, i_s\} \subseteq [n]$.
For such problems, an attractive option to use for the auxiliary structure $\Xi$ is a \emph {supersequence} of the regions corresponding to $S$; that is, $\Xi$ is a sequence of (possibly reccuring) elements of $\RR$ with the guarantee that, no matter where the true points $P$ lie in their regions, the sequence of elements of $P$ which we would obtain by replacing the regions in $\RR$ by their points will always contain $S(P)$ as a subsequence.
Then, after replacing the regions by their realization, we only need to determine for each element in the sequence whether it should be output or skipped, given the promise that the desired output is already in the correct order: a clean computational problem which we believe has intrinsic value independent of the preprocessing framework.

In this work, we explore this idea in arguably the most fundamental computational problem for which the output is a sequence of input elements: sorting, as well the two closely related problems of finding the smallest and largest gaps in (the sorted order of) a set of unsorted numbers.
We believe this approach has several benefits:
\ben{Somewhere in this section we should be clear that our solution is not "optimal" and make a forward ref to the remark I added. In particular below we need to make clearer the distinction between sequence and supersequence benefits.}
\begin {itemize}
  \item A sequence is arguably a simpler and more meaningful structure than the structures used in previous work;
  \item the same concept is potentially applicable to many different problems where the output is a sequence of input elements, e.g. shortest path in graphs or convex hulls in the plane~\cite{2d-arxiv}.
  \maarten {For future conference submission: discuss both that sorting is a useful subroutine (in addition to being interesting by itself), separately from the statement that the concept of supersequences itself is interesting in other contexts.}
  \item the more natural auxiliary structure makes our preprocessing and reconstruction algorithms also of independent interest, outside of the framework;
  \item the supersequence viewpoint also naturally allows us to incorporate sublinear reconstruction times in certain settings, something that was not possible in most earlier work and has only been explored more recently.
\end {itemize}

\subsection{Sequences and sorting}
Given a set of intervals $\RR$, our goal will be to produce a sequence $\Xi(\RR)$ that is ``close to being sorted'', no matter where the true points are. We are particularly interested in sequences that are guaranteed to be a {\em supersequence} of the true points.\footnote {As most work in the preprocessing model is on higher-dimensional problems, but we focus on one-dimensional problems, we will also use the term ``point'' to refer to values in $\R$ (e.g. 1-dimensional points), and we use the terms ``point'', ``number'', and ``value'' interchangebly.}
Such a sequence always exists, but of course its length will depend on the amount of overlap of the intervals. In the worst case, when all intervals overlap each other, all sorted orders are possible.

A {\em universal word} is a string that contains as subsequences all permutations of its distinct characters.
The question of constructing the shortest universal word was first posed by Knuth and attributed to Karp~\cite {10.5555/891957}.
The length of known shortest universal words are in \cite {oeis/A062714}.
The best known upper bound is quadratic and is due to Tan~\cite {tan2022skiplettersshortsupersequence}; an earlier and more expressable quadratic bound of $n^2 - 7/3n + 19/3$ was given by Radomirovic~\cite {Radomirovic2012ACO}.
As a further benefit, his proof is constructive and also gives an $O(n^2)$ algorithm to construct such a word.
  Uzna\'{n}ski shows that testing whether a given string is a universal word is coNP-complete~\cite {u2015permutationssupersequenceconpcomplete}.

Clearly, a sequence $\Xi$ of quadratic length defeats the purpose of preprocessing, so for intervals that overlap too much a supersequence is not a viable option. On the other hand, many classical results in the preprocessing framework assume bounded {\em ply} (maximum overlap of the regions; see Section~\ref {sec:definitions} for formal definitions), and we will show in this note that for sets of intervals with constant ply, supersequences of linear size always do exist.

\begin {figure}
 \begin{center}
  \includegraphics[width=\textwidth]{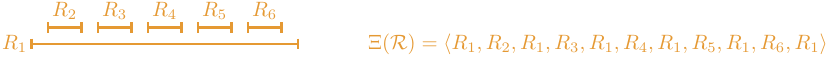}
 \end{center}
 \caption {A set of intervals $\RR$ of ply 2, and a supersequence that contains most elements only once, but one element a linear number of times.}
 \label {fig:intro-many}
\end {figure}

It is worth noting that the notion of a linear-size supersequence is flexible enough to allow {\em some} intervals to appear in the sequence very often, even a linear number of times (which in fact is sometimes necessary; see Figure~\ref {fig:intro-many}), as long as the total sequence length is not too large.

Once we have a supersequence, in the reconstruction phase, the intervals are replaced by the true locations of the points. This will now give us a sequence of points (with repeated instances) which comes with the promise that some subsequence of it is the sorted sequence of input points. We only need to find out which points belong to this subsequence.
The problem of finding a sorted sequence in a supersequence is closely related to the {\em longest increasing subsequence} problem. 
Fredman shows that the problem has a $\Omega (n \log n)$ lower bound~\cite{FREDMAN197529}. There is also a matching $O(n\log n)$ dynamic programming algorithm, which is commonly taught in advanced algorithms courses. 
However, our reconstruction problem has additional structure which makes it potentially easier.
We show in this note that on a Word RAM machine, it is indeed easier, and we can solve it in linear time. 
Whether the same is possible on a Real RAM is left as an open problem (see Section~\ref {sec:limitations}).

\subsection{Sublinear reconstruction}
Traditionally, the aim in the preprocessing model has been to achieve reconstruction times that are faster than computing a solution from scratch, with the understanding that there is a natural lower bound of $\Omega(n)$ time to reconstruct $S(P)$, since even replacing each region of $\RR$ with the corresponding point in $P$ will take linear time.
However, for problems in which the output is an ordered sequence of points, this reasoning is not entirely satisfactory, for two reasons:
\begin {enumerate}
  \item it is possible that not all elements of $P$ appear in $S(P)$; in this case it is not a priori clear that we need to spend time retrieving points that do not need to be output;
  \item even for those points that do appear in $S(P)$, in some applications one might be happy knowing just the order and not the exact points; if this order is already clear from $\RR$ there is again a priori no reason to believe we necessarily need to spend even $|S(P)|$ time in the reconstruction phase.
\end {enumerate}
Motivated by these observations, van der Hoog \etal~\cite {hkls-paip-19, hkls-pippf-22} have recently explored how the preprocessing model may be modified to allow for such sublinear reconstruction; 
for instance, for sorting they achieve reconstruction time proportional to the entropy of the interval graph, which can be less than $n$ if many intervals are isolated.\footnote {Note that to achieve reconstruction time proportional to the entropy, it is necessary to use a more powerful model than the traditional preprocessing framework, in which computation and retrieval can be mixed - i.e., one is able to retrieve a point, and depending on its location, inform how the computation continues and what other points should be retrieved.} 
In order to support sublinear reconstruction, van der Hoog \etal introduce an additional phase to the preprocessing framework, see Section~\ref{sec:1d-sub} for more details.

In contrast, a supersequence as an auxiliary structure fairly naturally allows for sublinear reconstruction, by either simply omitting elements of $\RR$ that are certainly not in the ouput, or marking those (sequences of) elements of $\RR$ that are guaranteed to be in the correct place in the sequence and need not be inspected during reconstruction.

\subsection{Smallest and largest gap}
In this note, we will demonstrate both types of sublinear reconstruction time by looking at the related problems of {\em sorting} and finding the {\em minimum} and {\em maximum gap} in a set of numbers (i.e., the smallest and largest difference between any two consecutive numbers in the sorted order).
Clearly, both the minimum and maximum gap can be computed in $O(n \log n)$ time by first sorting a set of numbers and then scanning the resulting list. For the smallest gap, a corresponding lower bound follows from element distinctness~\cite {Grigoriev}, while for the largest gap, the situation is more subtle: it has long been known that e.g. on a Real RAM where the use of the {\em floor} function is allowed, the largest gap can be computed in linear time~\cite {10.1007/3-540-62592-5_63,Gonzales}; on the other hand, current understanding is that the floor function should be disallowed for a well-defined Real RAM computation model~\cite {ehm22,kls-sehsr-16}.

In any case, clearly, any results for {\em sorting} in the preprocessing framework which achieve linear-time reconstruction of the sorted order, trivially imply corresponding results for both the smallest and largest gaps. However, when we aim for {\em sublinear} reconstruction, differences begin to appear.
Clearly, in the problem of sorting, all input elements also appear in the output, so sublinear reconstruction of the first type will not be possible. On the other hand, while all intervals that do not intersect any other intervals can be safely sorted without needing to retrieve their actual values, the same is not true for the gap problems. Hence, the problems are a priori no longer comparable.

\subsection{Organization}

The remainder of this paper is organized as follows.

In Section~\ref{sec:prelim} we introduce the necessary definitions to formally state our results.
We also state the main open problem that would close the gaps in our results (see below).

In Section~\ref{sec:1d-prep} we study the problem of preprocessing a set of intervals for sorting. Our main result here is an algorithm to produce a supersequence of intervals that is guaranteed to contain any instance in sorted order. When the intervals all have the same length, we show we can additionally achieve that the sequence is {\em \nice}.

In Section~\ref {sec:1d-rec}, we then consider the corresponding problem of reconstructing the sorted order from such a sequence. In particular, we show that it is always possible to recover the sorted order from a {\em \nice} supersequence in linear time on a Real RAM. For general sequences this is also possible in the Word RAM model of computation; whether the same is true on the Real RAM (and thus whether smoothness is necessary) is our main open question.

In Section~\ref {sec:sublinear}, we explore how the idea of using supersequences can be combined with sublinear reconstruction times, and develop our results for sorting, and for computing the smallest and largest gaps in a set of numbers.


\section{Preliminaries and Results}
\label {sec:prelim}

In this section we formally introduce the concepts and definitions necessary to state our problems and results.

\subsection {Definitions}
\label {sec:definitions}

\begin {definition}
  Let $\RR$ be a set of $n$ intervals.
  A {\em sorting-supersequence} of $\RR$ is a sequence $\Xi$ of (possibly reoccurring) intervals of $\RR$
  such that,
  for any point set $P \from \RR$, the sorted order of $P$ is a subsequence of $\Xi$.
\end {definition}

While this definition works for any set of intervals, it is easy to see that such a sequence is not useful for, e.g., a set of $n$ intervals that all coincide. Therefore, we will consider some natural restrictions on $\RR$.

\begin {definition}
The {\em ply} $\Delta$ of a set of intervals $\RR$ is the maximum over all points $p \in \R$ over the number of intervals that contain $p$:
  \[
    \Delta = \max_{p \in \R} \left| \left\{ I \in \RR \mid p \in I \right\} \right|
  \]
\end {definition}

We will focus on unit intervals parameterized by their ply $\Delta$ in this work.
Furthermore, we will consider an additional restriction on the sequences, and we will show that such a restricted sequence can both be computed on a set of  unit intervals of ply $\Delta$, and used to recover the sorted order; whether this restriction is necessary is left as an open problem (and may depend on the model of computation; refer to Section~\ref {sec:1d-rec}).

\begin {definition}
  A sequence of values $X=x_1,\ldots, x_n$ is {\em $(\alpha, \beta)$-\nice} if: 
  \begin {itemize}
    \item for any two elements $x_i, x_j \in X$ such that $i < j$ but $x_i > x_j$, we have $x_i - x_j \le \alpha$, and
    \hfill (distance property)
    \item for any interval $J \subset \R$ with $|J| > 1$, there are $\leq \beta |J|$ distinct elements from $X$ in $J$.
    \hfill (packing property)
  \end {itemize}
\end {definition}

We will also say a sorting-supersequence $\Xi$ is {\em $(\alpha, \beta)$-\nice} if the corresponding sequence of points (values) is {\em $(\alpha, \beta)$-\nice} for any $P \from \RR$.

\subsection{Problem statement \& results}

We can now formally state the problem of preprocessing intervals for sorting.

\begin{problem}\label{prob:1d}
Let $\RR$ be a set of $n$ intervals. 
Can we construct a supersequence $\Xi$ of $\RR$ such that given any $P\from \RR$ we can sort $P$ in $O(n)$ time by using $\Xi$?
\end{problem}

As mentioned, some restrictions on the intervals will be necessary for this to be possible.
We obtain the following specific preprocessing results.

\begin {theorem} \label{thm:mainunitintervals-preprocessing}
  Let $\RR$ be a set of $n$ unit intervals of ply $\Delta$.
  There exists a $(3, 2\Delta)$-\nice sorting-supersequence $\Xi$ of $\RR$ of size $O(n\Delta)$.
  Furthermore, such a sequence can be computed in $O(n (\Delta+\log n))$ time.
\end {theorem}

\begin {theorem} \label {thm:mainintervals-preprocessing}
  Let $\RR$ be a set of $n$ intervals of ply $\Delta$. There exists a sorting-supersequence $\Xi$ of $\RR$ of size $O(n\Delta^2)$. Furthermore, such a sequence can be computed in $O(n(\Delta^2 + log n))$ time.
\end {theorem}

Note that once we obtain a sorting-supersequence and replace the regions by their true point, it is no longer meaningful to distinguish between unit intervals and non-unit intervals, or to talk about the ply of the intervals. This is a main feature of using supersequences: it allows us to state generic reconstruction algorithms that are agnostic of how their input sequences were obtained.

We obtain two main reconstruction results.

\begin {theorem} \label{thm:mainunitintervals-reconstruction}
Let $\RR$ be a set of $n$ unit intervals. Given an $(\alpha, \beta)$-\nice sorting-supersequence $\Xi$ of $\RR$, and given a point set $P \from \RR$, the sorted order of $P$ can be computed in $O(|\Xi|\log(\alpha \beta))$ time on a Real RAM. 
\end{theorem}

\begin {theorem} \label{thm:mainunitintervals-reconstruction-wordram}
Let $\RR$ be a set of $n$ unit intervals. Given a sorting-supersequence $\Xi$ of $\RR$, and given a point set $P \from \RR$, the sorted order of $P$ can be computed in $O(|\Xi|)$ time on a Word RAM. 
\end{theorem}

Preprocessing and reconstruction algorithms may then be combined to obtain statements similar to that in the literature as required.
For instance,
when $\RR$ is a set of unit intervals of ply $\Delta$,  Theorem~\ref{thm:mainunitintervals-preprocessing} gives a $(3,2\Delta)$-\nice sorting-supersequence $\Xi$ of size $O(n\Delta)$. Combined with Theorem~\ref{thm:mainunitintervals-reconstruction}, we thus immediately have the following corollary, which matches known results for constant $\Delta$.

\begin{corollary}\label{cor:main}
Let $\RR$ be a set of $n$ unit intervals of ply $\Delta$.
There exists an auxiliary structure $\Xi$ such that for any $P \from \RR$, the sorted order of $P$ can be computed in $O(n \Delta\log(\Delta))$ time using $\Xi$.
\end {corollary}

\ben{Change above to specifically state the structure is a supersequence. Add remark below?}

\begin{remark}
If one allows $\Xi$ to be something other than a sorting-supersequence, then one can obtain better bounds than in Corollary~\ref{cor:main}. For example, let $\Xi$ be the sorted sequence of left interval endpoints. (So $\Xi$ is still a sequence, but not a  sorting-supersequence.) We can then sweep from left to right, maintaining a priority queue of the points corresponding to the currently intersected intervals, where the next event is the minimum of the next queue value and the next interval endpoint. As the intervals have ply $\Delta$, this queue has extraction cost $O(\log \Delta)$ and thus the total time to reconstruct the sorted order in $O(n\log \Delta)$.
\end{remark}

\subsection {Limitations \& open problems}
\label {sec:limitations}

In most of this work, we assume we are working in the Real RAM model of computation, which is consistent with earlier work in the preprocessing framework. (For a formal discussion of the Real RAM, refer to \cite{ehm22}.) However, in order to achieve linear reconstruction time in this model, we require our supersequences be \nice (defined below). 
In contrast, in Section~\ref{sec:wordram} we show we do not need the \nice restriction in the word RAM model, where universal hashing is allowed. 
Whether the restriction is necessary is the Ream RAM or not is unclear,
and we believe the answer to this question is of independent interest.
Hence, we formulate the following open problem:

\begin {problem} \label {prob:1d}
  Given a sequence of values $P$ in $\R$ (possibly with duplicates), and the guarantee that there exists a sorted subsequence $Q$ of $P$ such every distinct element in $P$ appears at least once in $Q$.
  Is it possible to sort $P$ in $o(n \log n)$ time?
\end {problem}

We note that this problem is a special case of the longest increasing subseqeunce problem, which has a $\Omega (n \log n)$ lower bound in the decision tree model~\cite{FREDMAN197529}.
However, the LIS lower bound may be more difficult, since there may be non-duplicated elements that are not in the output, whereas in our problem, if there are no duplicates, the sequence is already sorted, and we are done. Thus, if our problem has a superlinear lower bound, the presence of duplicates must play a crucial role, whereas this is not the case for LIS.

Another limitation of our work is the dependence on the ply $\Delta$. 
We assume that $\Delta$ is constant in this work (as is done in many other papers in the preprocessing framework), and under this assumption, our results are optimal.
We believe it is unlikely that the dependence on $\Delta$ can be improved when sticking rigidly to supersequences; yet, better dependence on $\Delta$ has been established in some prior work in the framework (though not explicitly for the problem of sorting). 
One promising avenue to explore in this direction is to relax the requirement that $\Xi$ is a supersequence of $P$, and instead require that $\Xi$ contains a subsequence which is ``close to being sorted'', e.g., has a small number of out-of-order pairs (where ``small'' must be some function of $\Delta$). Such a relaxation would allow generally for shorter sequences, but would need to be balanced by more complicated reconstruction algorithms.

\section{Preprocessing} \label {sec:1d-prep}

We will now consider in detail the problem of computing a sorting-supersequence from a set of intervals. We will first consider the special case of unit intervals, and then the more general case. In particular, when the intervals are unit intervals, we will be able to show additionally that the resulting sequences are {\em \nice}, which is not possible for the more general case.

\subsection {Preprocessing unit intervals}

We will use the following result.
\begin {lemma} [\cite{Radomirovic2012ACO}] \label {lem:shortestsuperseqeunce}
  The length of the shortest sequence on $m$ symbols which contains all permutations of $[1,m]$ as a subsequence is upper-bounded by $m^2$. Moreover, such a sequence of length $m^2$ can be computed in $O(m^2)$ time.
\end {lemma}

While Radomiroci\'c~\cite{Radomirovic2012ACO} does not discuss the running time, it is implicit in their construction. 
Ultimately we will apply it when $m$ is the ply, which we assume is bounded.
Using this result, we prove the following technical lemma.

\begin{lemma}\label{lem:interval-windows}
  Let $\RR$ be a set of $n$ unit intervals in $\R$ of ply $\Delta$.
  There is a $(3, 2\Delta)$-\nice sorting-supersequence $\Xi$ of $\RR$ with at most $4\Delta n$ intervals.
\end{lemma}

\begin {figure} \label {fig:interval-windows}
  \includegraphics [width=\textwidth] {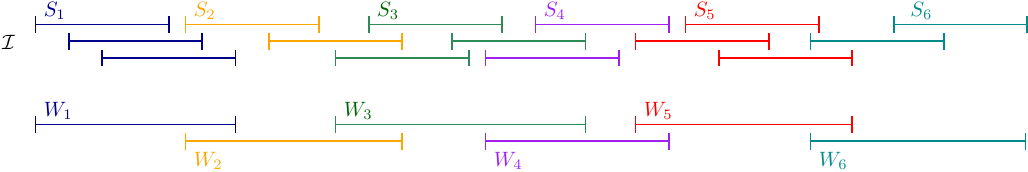}
  \caption {An example of a sequence of unit intervals of ply $\Delta=3$.}
\end {figure}

Note, to get a sequence of length $n \Delta^2$, we could simply apply Lemma~\ref {lem:shortestsuperseqeunce} to every cell of the subdivision induced by $\RR$. However, Lemma~\ref{lem:interval-windows} gives a more refined bound.

\begin {proof}
  Assume the intervals ${\RR} = \{I_1, \ldots, I_n\}$ are ordered and indexed from left to right (since they are unit intervals, this order is well-defined).
  Now, let $S_1 = \{I_i \mid I_i \cap I_1\neq \emptyset\}$ be the set of all intervals that intersect the leftmost interval (including the leftmost interval itself) and let $\eta_1 = \min \{i\mid I_i \cap I_1 = \emptyset\}$ be the index of the first interval not in $S_1$.
  We recursively define $S_{i+1}$ to be the set of all intervals that intersect $I_{\eta_{i}}$ and are not in $S_{i}$.
  Then, let $W_i = \bigcup_{I \in S_i} I$ be the window from the left endpoint of the leftmost interval in $S_i$ to the right endpoint of the rightmost interval in $S_i$.
  
  Now, define $S'_i$ to be the set of all intervals from $S_i$ unioned with all intervals from $S_{i+1}$ which intersect $W_i$.
  We claim that $|S'_i| \le 2 |S_i| - 1 \le 2\Delta - 1$. This is because any interval from $S_{i+1}$ intersecting $W_i$, that is not already in $S_i$, must contain the right endpoint of $W_i$, and as $\Delta$ is the ply, there can therefore be at most $\Delta - 1$ such intervals.
  We further claim that every pair of intervals that intersect each other is present in some set $S'_i$: indeed, an interval from $S_i$ cannot intersect an interval from $S_{i+2}$ since then the interval in $S_{i+2}$ would also intersect the leftmost interval of $S_{i+1}$ (but then it would be part of $S_{i+1}$).

  Let $k$ be the resulting number of $S_i'$ sets.
  We process these sets of intervals from left to right, and apply Lemma~\ref {lem:shortestsuperseqeunce} to $S'_i$ for $i = 1, \ldots, k$. Let the resulting $i$th sequence be $\Xi_i$. Then we simply concatenate the resulting sequences, to get a sequence $\Xi = \Xi_1 \cup \Xi_2 \cup \cdots \cup \Xi_{k}$.

 As argued above, $|S_i'|\leq 2\Delta$ and moreover $\sum_{i=1}^k |S_i'| \leq 2\sum_{i=1}^k |S_i|=2n$. By Lemma~\ref {lem:shortestsuperseqeunce}, $|\Xi_i|=|S_i'|^2$. Thus $|\Xi|$  is bounded by a sum of squares of values, i.e.\ the $|S_i'|$, which are each bounded by $2\Delta$ and sum to at most $2n$. It is not hard to see that such a sum of squares is maximized when $S_i'=2\Delta$ for all $i$ and $k= n/\Delta$. Thus the length of $\Xi$ is at most $n/\Delta \cdot 4\Delta^2 = 4n\Delta$.

  Now, let $P \from {\RR}$ be any realization, listed in sorted left to right order on the real line.  
  We divide the points in $P$ into $k$ groups: let $P_i \subset P$ be the set of all points that are contained in $W_i$ but not in $W_{i-1}$.
  Let $\Pi=\Pi(P)$ be the ordered string of intervals corresponding to the ordered points of $P$, and let $\Pi_i$ be the part of $\Pi$ corresponding to $P_i$.
 The string $\Pi_i$ is a substring of $\Xi_i$ (by Lemma~\ref {lem:shortestsuperseqeunce}).
 Therefore, the points in $P_i$ can be charged to the subsequence $\Xi_i$, thus proving $\Xi$ is a sorting-supersequence, as the $\Xi_i$ are concatenated in order.

  Finally, we need to argue that the resulting sequence $\Xi$ is $(3, 2\Delta)$-\nice.
  To argue that $\alpha = 3$, note that all intervals in $S_i$ intersect each other, and thus all intervals in $S'_i$ either intersect each other or intersect a common interval. In the first case, the distance between any two points in them is at most $2$; in the latter case they could have points at distance $3$ from each other.
  Moreover, the $S_i'$ are sorted on the real line (say by the left ends of their leftmost interval), and since we concatenate the $\Xi_i$ in this order, this bound for $\alpha$ on each $\Xi_i$ holds for $\Xi$ as well.
  To argue that $\beta = 2\Delta$, note that the largest number of unit intervals that can intersect an interval of length $J$ is $(|J| + 1) \cdot \Delta$, 
  where $(|J| + 1) \cdot \Delta\leq 2\Delta |J|$ when $|J|>1$.
\end {proof}

We are now ready to finish the preprocessing algorithm.

\begin {proof} [Proof of Theorem~\ref{thm:mainunitintervals-preprocessing}]
Observe that Lemma~\ref{lem:interval-windows} establishes the existence of the desired sorting-supersequence. Moreover, the proof of Lemma~\ref{lem:interval-windows} is constructive. First, sort the $n$ unit intervals in $O(n\log n)$ time, yielding the $S_i$ sets. $S_i'$ can then be computed (using this sorted ordering) from $S_i$ in $O(\Delta)$ time, and thus $O(\Delta n)$ time over all $i$. Finally, Lemma~\ref{lem:shortestsuperseqeunce} states that the time to compute $\Xi_i$ from $S_i'$ is proportional to $|\Xi_i|$ (which itself is $|S_i'|^2$). 
Thus $\Xi$ is computed from the $S_i'$ in $O(\Delta n)$ time overall, as it is the concatenation of the $\Xi_i$, and Lemma~\ref{lem:interval-windows} already argued $|\Xi|\leq 4\Delta n$.
\end {proof}

\subsection {Preprocessing general intervals}

Now, consider a general set of intervals $\RR$ of ply $\Delta$. In this case, we will generally not be able to produce a {\nice} sorting-supersequence: indeed, as illustrated in Figure~\ref {fig:intro-many}, there may be intervals whose point could appear arbitrarily far from any occurance in the sequence. Nonetheless, we can still obtain a sorting-supersequence of linear size.

\begin {proof} [Proof of Theorem~\ref{thm:mainintervals-preprocessing}]
  Since the number of intervals is $n$ and each interval has only two endpoints, the complexity of the subdivision induced by $\RR$ is still only linear.
  Hence, we can apply Lemma~\ref {lem:shortestsuperseqeunce} to every cell of this subdivision.
  Since each cell intersects at most $\Delta$ intervals, whis results in $n$ sequence of length $\Delta^2$ each, or, when concatenated, a single sequence of length $n\Delta^2$.
\end {proof}

\section {Reconstruction} \label {sec:1d-rec}

We now consider the problem of recovering the sorted order of a set of values from a supersequence as computed in Section~\ref {sec:1d-prep}. 
In this section we will consider this problem in the classical preprocessing framework, where the goal is to achieve linear-time reconstruction.
We show that for \nice sequences this is achievable in the Real RAM, while on the Word RAM we do not require the sequences to be \nice.

\label {sec:1d-rec-standard}

For the reconstruction phase, we are now given a sorting-supersequence $\Xi$ of $\RR$, and additionally a set of points (values) $P \from \RR$, and the goal is to efficiently compute the sorted order of $P$ using $\Xi$. (In particular, in linear time when $\Xi$ is $(O(1),O(1))$-\nice.)

We will assume that the regions in $\Xi$ come as pointers to the original regions in $\RR$, and that our sets $\RR$ and $P$ are equipped with bidirectional pointers.

The natural approach to the reconstruction problem is to replace each region in $\Xi$ by its corresponding point in $P$. This results now in a sequence $X$ of numbers, which comes with a guarantee that the sorted order is a subsequence of $X$. 

\begin {definition}
Let $X=x_1, \ldots, x_n$ be a sequence of $n$ numbers (possibly with duplicate numbers), which contains $k\leq n$ unique numbers. Then we call $X$ a \emph{hidden sorted} sequence if it contains a subsequence of length $k$ with all $k$ distinct numbers in sorted order.  
\end {definition}

First, we show that given such a sequence $X$, we can find the sorted subsequence in linear time when universal hashing is allowed. Subsequently, we show even without universal hashing linear time is possible when the sequence is appropriately smooth.

\subsection{Sorting hidden sorted sequences in linear time in the word RAM model}
\label {sec:wordram}

We first show that there is a natural greedy algorithm to sort hidden sorted sequences, which runs in expected linear time assuming universal hashing. 

\begin {lemma}\label{lem:hash}
Given a hidden sorted sequence $X=x_1, \ldots, x_n$ with $k$ unique numbers, there is a greedy algorithm which outputs a sorted subsequence of length $k$ containing all $k$ distinct numbers. This algorithm runs in $O(n)$ expected time if universal hashing is allowed, and otherwise it runs in $O(n \log k)$ time.
\end {lemma}

\begin {proof}

Create a stack which initially contains just $x_1$. Now process the items in the sequence in order (starting at $x_2$). Let $x_i$ be the current item being processed.
If $x_i$ exists in the stack already then discard it. Otherwise, pop items off the stack until the top item is less than $x_i$, at which point we push $x_i$ onto the stack.%
\footnote{Note that the above algorithm is oblivious of the value $k$.}

Let $x_{i_1},\ldots,x_{i_k}$ be any subsequence of $X$ containing all $k$ unique items in sorted order. To argue correctness, we prove by induction that after $x_{i_j}$ has been processed by our algorithm, the bottom of the stack contains the $j$ smallest (unique) items in sorted order. (For $j<k$ there may be additional items above on the stack, but for $j=k$ this is not possible as by construction the stack contains sorted unique items.) For the base case when $j=1$, the claim trivially holds as then the smallest item has been processed (either when $x_{i_j}$ is processed or earlier), at which point the algorithm will place it at the bottom of the stack, and the algorithm will never remove it. Now consider some $j>1$. By induction, after $x_{i_{j-1}}$ was processed, the stack contains the $j-1$ smallest (unique) items in sorted order at the bottom, and these can never be removed by the algorithm. Thus as $x_{i_j}$ is the next smallest item, the same logic as in the base case implies that after $x_{i_j}$ is processed, it will be on the stack immediately after the first $j-1$ items, and will never be removed.

Assuming we store the current stack items in a hash table, it takes expected constant time to check if $x_i$ is already in the stack. 
Thus we get a linear time algorithm, since in each time step we are either inserting or deleting from the stack, and each item is inserted at most once. If hashing is not allowed, we can instead use a binary tree for the items in the stack (whose size is at most $k$), and thus the running time is $O(n\log k)$.
\end {proof}

With this result in place, we can now prove Theorem~\ref{thm:mainunitintervals-reconstruction-wordram}.

\begin {proof} [Proof of Theorem~\ref{thm:mainunitintervals-reconstruction-wordram}]
If $P\from\RR$ is a realization of $\RR$ and $\Xi$ is a sorting-supersequence of $\RR$, then, by definition, if we replace the regions from $\Xi$ by their corresponding points in $P$, we obtain a hidden sorted sequence of the same length as $\Xi$. Thus, using Lemma~\ref{lem:hash}, we can find the sorted order of $P$ in $O(|\Xi|)$ time, as required. 
\end {proof}

\subsection{Sorting \nice sequences}

For our purposes, it is sufficient to recover the sorted subsequence from a hidden sorted sequence that additionally is {\em $(\alpha, \beta$)-\nice}. We assume $\alpha$ and $\beta$ are constants, so the goal is to sort the sequence in time of the form $O(f(\alpha, \beta) \cdot n)$.

Note that if we convert the supersequence $\Xi$ as constructed in Section~\ref {sec:1d-prep} to their corresponding values in $P$, we get a sequence that is $(3, 2\Delta)$-\nice.

\begin {lemma} \label {lem:linearrecovery}
  Given a $(\alpha, \beta)$-\nice hidden sorted sequence $X=x_1, \ldots, x_n$ with $k$ unique numbers
  we can output a sorted subsequence of length $k$ containing all $k$ distinct numbers in $O(n \log(\alpha \beta))$ time in the Real RAM model.
\end {lemma}

\begin{proof}
Recall the algorithm of Lemma~\ref{lem:hash}, which greedily processes the items in $X$ in order, discarding $x_i$ if it has already occurred, and otherwise popping items off the stack until the top item is less than $x_i$, at which point we push $x_i$ onto the stack.  
We already argued this algorithm is correct, though its running time depends on how long it takes to determine whether $x_i$ has already occurred. 

The claim is that if $x_i$ has already occurred, then is must be one of the largest $\alpha \beta$ items seen so far. Assuming this claim is true, we can implement this algorithm in $O(n\log(\alpha\beta))$ time, as maintaining a balanced binary tree with the largest $\alpha\beta$ items seen, takes $O(\log(\alpha\beta))$ time per round (i.e.\ $O(n\log(\alpha\beta))$ overall), and the remaining stack operations are linear time overall. 

To prove the claim, fix any $x_i$, and let $x_j$ be the largest value item such that $j<i$. If $x_i>x_j$, the $x_i$ has not been seen before, and the claim trivially holds. So suppose that $x_i\leq x_j$. Since $X$ is $(\alpha,\beta)$-\nice, we have that $x_j-x_i\leq \alpha$. Thus if $x_i$ has occurred already, it must be within an interval of length $\alpha$ (ending at $x_j$), and there can be at most $\alpha \beta$ distinct points in such an interval, by the packing property of $(\alpha,\beta)$-\nice sequences, thus implying the claim.
\end{proof}

\begin {proof} [Proof of Theorem~\ref{thm:mainunitintervals-reconstruction}]
If $\Xi$ is an {\em $(\alpha, \beta)$-\nice} sorting-supersequence of $\RR$, then by definition for any $P\from\RR$, if we replace the regions from $\Xi$ by their corresponding points in $P$, then it results in an $(\alpha,\beta)$-\nice hidden sorted sequence. Thus Lemma~\ref{lem:linearrecovery}, directly states we can recover the sorted order of $P$ in $O(|\Xi| \log(\alpha\beta))$ time. 
\end {proof}

\section{Sublinear Reconstruction}
\label {sec:sublinear}

In this final section, we explore the suitability of using supersequences as auxiliery structures when combined with the recent trend of sublinear reconstruction.

The preprocessing framework was originally developed with linear-time reconstruction as the ultimate goal: since, during the reconstruction phase, as we need to spend linear time to replace each region in $\RR$ by the corresponding point in $P$, we cannot hope to be faster.
Van der Hoog~\etal~\cite{hkls-paip-19,hkls-pippf-22} argue that this viewpoint may be limiting, and that in some applications we may be content with an output structure that still contains some regions, as long as we are guaranteed that the true points in those regions are combinatorially in the correct location. They argue that the preprocessing framework may be adapted to allow for this behaviour, at the cost of introducing another phase: the reconstruction phase is formally separated into a first subphase (that can take sublinear time), in which the auxiliary structure $\Xi$ is transformed into another structure $\Xi'$ which is combinatorially equivalent to the desired output $S(P)$, and a second subphase in which $\Xi'$ is actually transformed into $S(P)$ in linear time, if so desired.

We show that the same behavior can be quite naturally obtained when the auxiliary structure $\Xi$ is a supersequence.

\begin {figure}
 \begin{center}
  \includegraphics[width=\textwidth]{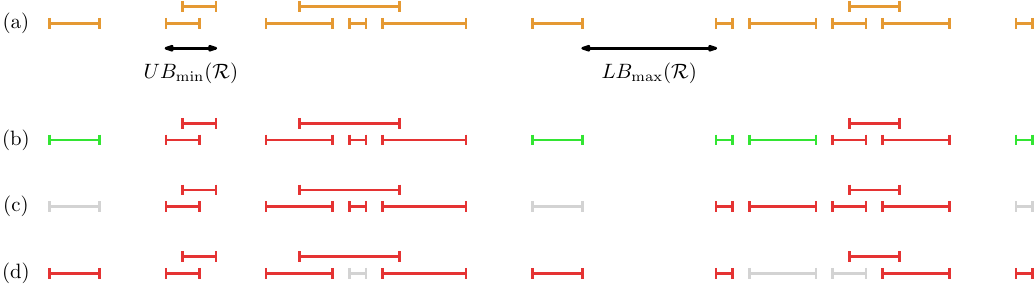}
 \end{center}
 \caption 
 { (a) A set of intervals $\RR$ of ply 2, with the upper bound on the smallest gap and the lower bound on the largest gap indicated.
   (b) The marked intervals (red) and unmarked intervals (green) for sorting.
   (c) The marked intervals (red) and omitted intervals (grey) for the smallest gap.
   (d) The marked intervals (red) and omitted intervals (grey) for the largest gap.
 }
 \label {fig:sublinear}
\end {figure}

\subsection{Sublinear reconstruction for sorting}
\label {sec:1d-sub}

As a pure supersequence of the input contains no information on which elements are necessary and which are not, its applicability is limited; nonetheless, even in this purest interpretation, we may obtain sublinear reconstruction for certain problems by simply omitting un-needed elements (see Section~\ref {sec:gaps}). For the purpose of sorting, however, this is not possible, and we require an annotated version of a sequence of input elements, in which is indicated which intervals will need to be retrieved and which not.
A {\em marked} sorting-supersequence is a sorting-supersequence in which all the intervals whose value needs to be retrieved are marked.
Hence, a standard sorting-supersequence is a marked supersequence in wich all items are marked. To achieve maximum benefit, we want to mark as few items as possible, since all non-marked items may be skipped during reconstruction.

\begin {definition}
  A {\em marked} sorting-supersequence $\Xi$ of $\RR$ is a sequence of (possibly recurring) intervals of $\RR$, some of which may be marked, and such that, for any point set $P \from \RR$, the sorted order of $P$ is a subsequence of $\Xi$, that contains all non-marked items of $\Xi$.
  Furthermore, each marked element in such a sequence stores a pointer to the next marked element.
\end {definition}

In the case of sorting, intervals are marked if the position in the sorted order depends on the location of the point in the interval. This happens exactly when the interval overlaps at least one other interval. An example is given in Figure~\ref {fig:sublinear} (b).

We adapt the preprocessing algorithm as follows. In $O(n \log n)$ time, we sort the intervals (by their left endpoint) as before. Now, in linear time, we scan through the intervals and detect any intervals that overlap at least one other interval. We mark these intervals, and collect clusters of marked intervals that are separated by at least one non-marked interval. We apply Lemma~\ref {lem:interval-windows} to each cluster separately, and concatenate the resulting sequences. 

In order to be able to quickly skip over unmarked intervals during the reconstruction phase, we can equip the resulting sorting-supersequence with pointers from each item to the next marked item. During the reconstruction phase we then also apply Lemma~\ref {lem:linearrecovery} separately to each subsequence of consecutive marked items, in time propertional to the total number of marked items in the list. This results in a mixed sequence of sorted points and marked intervals, which is guaranteed to be in the correct sorted order (this list would be called $\Xi'$ in the terminology of~\cite{hkls-paip-19}).

\subsection {Sublinear reconstruction for the smallest or largest gap}
\label {sec:gaps}

When the goal is not to sort the points, but to find the smallest gap, we may additionally potentially discard some input intervals of which we are sure they can never contribute to the desired quantity. To this end, we determine an upper bound on the size of the smallest gap, and a lower bound on the size of the largest gap, based on the intervals.

\begin {definition}
  The {\em smallest gap upper bound}, $UB_{\min}({\RR})$, of a set of intervals $\RR$, is the minimum over all pairs of intervals in $\RR$ of the maximum distance between any pair of points from those intervals.
\end {definition}

\begin {definition}
  The {\em largest gap lower bound}, $LB_{\max}({\RR})$, of a set of intervals $\RR$, is the maximum distance from the right endpoint of an interval to the left endpoint of the next interval in the sorted sequenece of all interval endpoints.
\end {definition}

These bounds can easily be computed in $O(n \log n)$ time by sorting the endpoints of the intervals in $\RR$. Then, using these bounds, we can classify the intervals into those which could potentially contribute to the smallest or largest gap (marked intervals), and those which cannot (omitted intervals); see Figure~\ref {fig:sublinear} (c) and (d). Note that for these problems there are generally no unmarked intervals. 

Testing which intervals should be omitted for the smallest gap is a simple scan: exactly those intervals that are further than $UB_{\min}({\RR})$ away from both neighbours can be omitted. This can be done in linear time, independent of the ply of $\RR$.

Testing which intervals should be omitted for the largest gap is slightly more complicated, but can also be done in linear time after separetely sorting the left and right endpoints of the intervals: an interval $I$ can be omitted if and only if the left endpoint of $I$ is closer than $LB_{\max}({\RR})$ to the next right endpoint of any other interval, and the right endpoint of $I$ is closer than $LB_{\max}({\RR})$ to the previous left endpoint of any other interval.

Once the unneeded intervals are omitted, we may proceed by preprocessing the remaining intervals as in Section~\ref {sec:1d-prep}. In the reconstruction phase, we then first obtain the sorted order of the points, and from this compute the smallest or largest gap, in time proportional to the number of marked intervals.

While these results are, admittedly, rather straightforward, it is exactly their simplicity \maarten {we might want to tone down the "simplicity" for anything other than another SOSA (or the corresponding ESA track?) submission - in particular in an arxiv version?} that showcases the potential of supersequences being used in the analysis of more complex computational problem on which the output is a subsequence of the input, such as e.g. shortest paths, convex hulls, orienteering problems, and more, in the context of the preprocessing framework for dealing with uncertain data.


\bibliographystyle{plainurl}
\bibliography{refs,geom,moregeom,maarten}

@techreport
{ Gonzales
, author      = "T. Gonzales"
, title       = "Algorithms on sets and related problems"
, institution = "Department of Computer Science, University of Oklahoma"
, year        = 1975
}

@InProceedings{10.1007/3-540-62592-5_63,
author="G{\'o}mez, Francisco
and Ramaswami, Suneeta
and Toussaint, Godfried",
editor="Bongiovanni, Giancarlo
and Bovet, Daniel Pierre
and Di Battista, Giuseppe",
title="On removing non-degeneracy assumptions in computational geometry",
booktitle="Algorithms and Complexity",
year="1997",
publisher="Springer Berlin Heidelberg",
address="Berlin, Heidelberg",
pages="86--99",
isbn="978-3-540-68323-0"
}

@inproceedings
{ kls-sehsr-16
, category    = "Terrain Modelling"
, paper       = "Space-Efficient Hidden Surface Removal"
, author      = "Frank Kammer and Maarten L{\"o}ffler and Rodrigo I. Silveira"
, title       = "Space-Efficient Hidden Surface Removal"
, class       = "Wet non"
, booktitle   = "MASSIVE'16"
, venue       = "MASSIVE"
, year        = "2016"
, affil       = "Utrecht"
, slides      = "massive_space_fish.pdf"
}

@article{ehm22,
author = {Erickson, Jeff and Hoog, Ivor and Miltzow, Tillmann},
year = {2022},
month = {04},
pages = {FOCS20-102},
title = {Smoothing the Gap Between NP and ER},
journal = {SIAM Journal on Computing},
doi = {10.1137/20M1385287}
}

@misc{tan2022skiplettersshortsupersequence,
      title={Skip Letters for Short Supersequence of All Permutations}, 
      author={Oliver Tan},
      year={2022},
      eprint={2201.06306},
      archivePrefix={arXiv},
      primaryClass={math.CO},
      url={https://arxiv.org/abs/2201.06306}, 
}

@article{Grigoriev,
  author = "Grigoriev, Dima and Karpinski, Marek and auf der Heide, Friedhelm Meyer and Smolensky, Roman",
  year = "1996",
  title = "A lower bound for randomized algebraic decision trees",
  journal = "Computational Complexity",
  pages = "357--375",
  volume = "6",
  number = "4",
  url = "https://doi.org/10.1007/BF01270387",
  doi = "10.1007/BF01270387"
}

@article{held2008triangulating,
  title={ {Triangulating Input-constrained Planar Point Sets}},
  author={Held, Martin and Mitchell, Joseph},
  journal={Information Processing Letters (IPL)},
  year={2008}
}

@article{devillers2011delaunay,
  title={ {Delaunay Triangulation of Imprecise Points: Preprocess and Actually get a Fast Query Time}},
  author={Devillers, Olivier},
  journal={Journal of Computational Geometry (JoCG)},
  year={2011}
}

@article{FREDMAN197529,
title = {On computing the length of longest increasing subsequences},
journal = {Discrete Mathematics},
volume = {11},
number = {1},
pages = {29-35},
year = {1975},
issn = {0012-365X},
doi = {https://doi.org/10.1016/0012-365X(75)90103-X},
url = {https://www.sciencedirect.com/science/article/pii/0012365X7590103X},
author = {Michael L. Fredman}
}

@misc{u2015permutationssupersequenceconpcomplete,
      title={All Permutations Supersequence is coNP-complete}, 
      author={Przemysław Uznański},
      year={2015},
      eprint={1506.05079},
      archivePrefix={arXiv},
      primaryClass={cs.CC},
      url={https://arxiv.org/abs/1506.05079}, 
}

@article{Radomirovic2012ACO,
  title={A Construction of Short Sequences Containing All Permutations of a Set as Subsequences},
  author={Sasa Radomirovic},
  journal={Electron. J. Comb.},
  year={2012},
  volume={19},
  pages={31},
  url={https://api.semanticscholar.org/CorpusID:8872002}
}

@misc{oeis/A062714
, title = "Online encyclopedia of integer sequences, sequence A062714"
, url   = "http://oeis.org/A062714"
, note  = "Accessed: 2025-07-08"
}

@techreport{10.5555/891957,
author = {Chvatal, Vaclav and Klarner, David A. and Knuth, Donald E.},
title = {Selected combinatorial research problems.},
year = {1972},
publisher = {Stanford University},
address = {Stanford, CA, USA}
}

@article{buchin2011delaunay,
  title= {{Delaunay Triangulations in {$O (\mathit{sort}(n))$} Time and More}},
  author={Buchin, Kevin and Mulzer, Wolfgang},
  journal={Journal of the ACM (JACM)},
  year={2011}
}

@article{van2010preprocessing,
  title={ {Preprocessing Imprecise Points and Splitting Triangulations}},
  author={van Kreveld, Marc and L{\"o}ffler, Maarten and Mitchell, Joseph},
  journal={SIAM Journal on Computing (SICOMP)},
  year={2010}
}

@article{loffler2013unions,
  title={ {Unions of Onions: Preprocessing Imprecise Points for Fast Onion Decomposition}},
  author={L{\"o}ffler, Maarten and Mulzer, Wolfgang},
  journal={Journal of Computational Geometry (JoCG)},
  year = 2014
}

@article{loffler2010delaunay,
  title={ {Delaunay Triangulation of Imprecise Points in Linear Time after Preprocessing}},
  author={L{\"o}ffler, Maarten and Snoeyink, Jack},
  journal={Computational Geometry},
  year={2010}
}

@article{ezra2013convex,
  title={ {Convex Hull of Points Lying on Lines in $o(n\log n)$ Time after Preprocessing}},
  author={Ezra, Esther and Mulzer, Wolfgang},
  journal={Computational Geometry},
  year={2013}
}

@article{BLMM11,
  title={ {Preprocessing imprecise points for Delaunay triangulation: Simplified and extended}},
  author={Buchin, Kevin and L{\"o}ffler, Maarten and Morin, Pat and Mulzer, Wolfgang},
  journal={Algorithmica},
  year={2011}
}

@inproceedings
{ hkls-paip-19
, author      = "Ivor van der Hoog and Irina Kostitsyna and Maarten L{\"o}ffler and Bettina Speckmann"
, title       = "Preprocessing Ambiguous Imprecise Points"
, booktitle   = "Proc. 35th Symposium on Computational Geometry"
, venue       = "SoCG"
, year        = "2019"
}

@inproceedings
{ hkls-pippf-22
, paper = "Preprocessing Imprecise Points for the Pareto Front"
, title = "Preprocessing Imprecise Points for the Pareto Front"
, author      = "Ivor van der Hoog and Irina Kostitsyna and Maarten L{\"o}ffler and Bettina Speckmann"
, booktitle   = "Proc. 32nd Symposium on Discrete Algorithms"
, year        = "2022"
}

@misc{2d-arxiv,
      title={Preprocessing Disks for Convex Hulls, Revisited}, 
      author={Maarten Löffler and Benjamin Raichel},
      year={2025},
      eprint={2502.03633},
      archivePrefix={arXiv},
      primaryClass={cs.CG},
      url={https://arxiv.org/abs/2502.03633}, 
}

\end{document}